\documentclass[a4paper,12pt]{article}
\usepackage{amssymb,amsthm,comment}
\usepackage[lined,algonl,boxed]{algorithm2e}
\usepackage{tabularx}
\usepackage{latexsym}
\usepackage{amsmath}
\usepackage{amssymb}

\newtheorem{lem}{Lemma}
\newtheorem{thm}{Theorem}

\newtheorem{defn}{Definition}

\title{\textbf{Maximum Matchings via Glauber Dynamics}}

\author{Anant Jindal \thanks{Laxmi Niwas Mittal Institute of Information Technology, India. {\tt anantjindal1@gmail.com}} 
\and Gazal Kochar \thanks{Laxmi Niwas Mittal Institute of Information Technology, India.  {\tt gkochar@gmail.com}} 
\and Manjish Pal \thanks{Indian Institute of Technology Gandhinagar, India. {\tt manjish\_pal@iitgn.ac.in}} 
}

\begin{document}

\date{}

\maketitle

\abstract{ 
In this paper we study the classic problem of computing a maximum cardinality matching in general graphs $G = (V, E)$.
This problem has been studied extensively more than four decades. The best known algorithm for this
problem till date runs in $O(m \sqrt{n})$ time due to Micali and Vazirani \cite{MV80}. Even for general bipartite graphs
this is the best known running time (the algorithm of Karp and Hopcroft \cite{HK73} also achieves this bound).  
For regular bipartite graphs one can achieve an $O(m)$ time algorithm which, following a series of papers, has been recently improved to $O(n \log n)$ by 
Goel, Kapralov and Khanna (STOC 2010) \cite{GKK10}. In this paper we present a randomized algorithm based on the Markov Chain
Monte Carlo paradigm which runs in $O(m \log^2 n)$ time, thereby obtaining a significant improvement over \cite{MV80}.\\

We use a Markov chain similar to the \emph{hard-core model} for Glauber Dynamics with \emph{fugacity} parameter $\lambda$, which is used
to sample independent sets in a graph from the Gibbs Distribution \cite{V99}, to design a faster algorithm for finding
maximum matchings in general graphs. Motivated by results which show
that in the hard-core model one can prove fast mixing times (for e.g. it is known that 
for $\lambda$ less than a critical threshold the mixing time of the hard-core model is $O(n \log n)$ \cite{MSW04}, 
we define an analogous Markov chain (depending upon a parameter $\lambda$) 
on the space of all possible partial matchings of a given graph $G$, for which the probability 
of a particular matching $M$ in the stationary follows the Gibbs distribution
which is:
\[
 \displaystyle \pi(M) = \frac{\lambda^{|M|}}{ \sum_{x \in  \Omega} \lambda^{|x|} }   
 \]
where $\Omega$ is the set of all possible matchings in $G$. \\ 

We prove upper and lower bounds on the mixing time of this Markov chain. Although our Markov chain
is essentially a simple modification of the one used for sampling independent sets from the Gibbs distribution,
their properties are quite different. Our result
crucially relies on the fact that the mixing time of our Markov Chain is independent of $\lambda$,
a significant deviation from the recent series of works \cite{GGSVY11,MWW09, RSVVY10, S10, W06} which achieve computational
transition (for estimating the partition function) on a threshold value of $\lambda$. 
As a result we are able to design a randomized algorithm which runs in $O(m\log^2 n)$ time that provides a
major improvement over the running time of the algorithm due to Micali and Vazirani. Using
the conductance bound, we also prove that mixing takes $\Omega(\frac{m}{k})$ time where $k$ is
the size of the maximum matching.  

}

\newpage

\section{Introduction}
Given an unweighted undirected graph $G = (V,E)$ with $|E| = m$ and $|V| = n$, a matching $M$ is 
a set of edges belonging to $E$ such that no two edges in $M$ are incident on a vertex.  If there is a
matching of size $n/2$ (for $n$ even), then it is called a \emph{perfect matching}. The Maximum Matching problem
is to find the maximum sized matching in a given graph. The computational complexity of this problem has
been studied extensively for more than four decades starting with an algorithm of Edmonds.

\subsection{General Graphs}
Edmonds's celebrated paper `Paths, Trees and Flowers' \cite{E65} was the first to give an efficient algorithm (also
called the \emph{blossom shrinking algorithm}) for finding maximum matching
in general graphs. This algorithm can be implemented in $O(n^4)$ time. The running time was
subsequently improved in a number of papers \cite{G73, KM74, L76}. All these papers were variants of Edmonds algorithm.
Even and Kariv \cite{GK82} obtained an improvement to $O(n^{2.5})$ which
was improved by Micali and Vazirani \cite{MV80} who gave an $O(m \sqrt{n})$ time algorithm for the problem by a careful handling 
of blossoms. This is the best known algorithm for finding maximum matchings in general bipartite graphs. 

\subsection{Bipartite Graphs}
For bipartite graphs, the problem can easily be solved using 
the max-flow algorithm by Ford and Fulkerson, an algorithm usually taught in an undergraduate algorithms course \cite{KT09},
which has a running time of $O(mn)$. The first algorithm for this problem was given by Konig \cite{K16}. 
Hopcroft and Karp \cite{HK73} gave an algorithm that runs in $O(m \sqrt{n})$ time. This
algorithm is an exact and deterministic algorithm. The problem becomes significantly simpler for regular bipartite graphs. 
In a \emph{d-regular} bipartite graph every vertex has degree $d$. When $d$ is a power of 2, Gabor and Kariv were able to achieve an $O(m)$
algorithm. After significant efforts, the ideas used there were used by Cole, Ost and Schirra \cite{COS01} to
obtain a get an $O(m)$ algorithm for general $d$. \\      		      

In a recent line of attack by Goel, Kapralov and Khanna \cite{GKK09, GKK09'}, the authors were able to use sampling 
based methods to get improved running time. In the most recent paper they were able to achieve a 
running time of $O(n \log n)$ for $d$-regular graphs \cite{GKK10}. Their algorithm performs an 
appropriately truncated random-walk on a modified graph to successively find augmenting path.

\section{Our Results}
In this paper we give a Markov Chain Monte Carlo algorithm for finding a maximum matching 
in general bipartite graphs. Our algorithm is in the spirit similar to \cite{GKK10} which 
also is a `truncated random walk' based algorithm, however the stationary distribution 
of the underlying Markov Chains in their case is different from ours. 
Inspired from the hard-core model with fugacity parameter $\lambda$ 
of sampling independent sets from graphs we define a similar Markov Chain over the  space of 
all possible partial matching such that its stationary distribution $\pi(\cdot)$ is the Gibbs distribution, ie. given
a matching $M$ its probability $\pi$ is
\[
 \displaystyle \pi(M) = \frac{\lambda^{|M|}}{ \sum_{\sigma \in  \Omega} \lambda^{|\sigma|} }   
 \]
where $\Omega$ is the set of all possible matchings in $G$. Notice that $\pi(M)$ is maximum for
maximum matchings and if $\lambda$ is a significantly large number, $\pi(M)$ tends to 1 for
maximum matchings. 
Our algorithm is extremely simple (being a standard in the MCMC paradigm). Starting from a fixed
matching we start a random walk in $\Omega$ according to the underlying graph $\tilde {G}$ of the Markov
chain. After $T_{mix}$ (mixing time) steps the distribution reached by the algorithm is roughly the same  as the
Gibbs distribution. More formally, the variation distance of ${\cal D}^{t}$ (the distribution after 
$t$ steps) from the Gibbs distribution is less than $\frac{1}{2e}$. This just leaves the task
of proving an upper-bound on the mixing time of the Markov Chain, for which we resort to the 
Coupling Method introduced by Bubley and Dyer \cite{BD97}. We define a metric $\Phi(\cdot, \cdot)$ and
apply the bound from \cite{BD97}, to show that the Markov Chain mixes in $O(m \log n)$ time.    
We also use the conductance method to show that the mixing time will be at least $\Omega(\frac{m}{k})$
where $k$ is the size of the maximum matching. Thus upto logarithmic factors the bounds are same 
when $k$ is small (at most polylogarithmic). 
The main result of our paper can be concisely written as follows:

\begin{thm}[\textbf{Main}]
There exists a randomized algorithm which given a graph $G = (V,E)$ with $|V| = n$ and $|E| = m$ 
finds a maximum matching in $O(m \log^2 n)$ time with high probability.
\end{thm}

\textbf{Important remark: It has been pointed to us independently by Yuval Peres, Jonah Sherman, Piyush 
Srivastava and other anonymous reviewers that the coupling used in this paper doesn't have 
the right marginals because of which the mixing time bound doesn't hold, and also the main
result presented in the paper. We thank them for reading the paper with interest and promptly pointing
out this mistake. }

\subsection{Organization}
The paper is organized as follows: in Section \ref{a} we give a brief description of the basic
idea and technique that are underlying our algorithm. Subsequently in Section \ref{Pre} we give
an overview of the MCMC paradigm which includes basic preliminaries and definitions regarding
Markov Chain and Mixing. Section 7 is devoted to the details of the chain being used by us and 
proving that indeed it has the desired properties. We then prove
upper and lower bounds on its mixing time using the coupling method and conductance argument
in Section 8 and Section 9 respectively. We end the paper with a conclusion and some open problems.

\section{The Idea: Maximum Matchings via Glauber Dynamics} \label{a}

Our ideas are inspired mainly from the results in the Glauber
dynamics of the hard-core model of fugacity parameter $\lambda$, for
sampling independent sets from the Gibbs distribution. According to the 
Gibbs distribution the probability of an independent set $I$ is given by 
\[
 \displaystyle {\cal G}(I) = \frac{\lambda^{|I|}}{ \sum_{\rho \in  \Omega} \lambda^{|\rho|} }   
 \]
where $\Omega$ is the set of all possible independent sets in $G$ and 
${\cal Z} = \sum_{\rho \in  \Omega} \lambda^{|\rho|}$ is also called the \emph{partition function}.
Clearly for $\lambda = 1$ the partition function value is same as the number of independent sets in 
the graph (computing which is a \#P-Hard problem). \\
In the hard-core model given a particular configuration $\sigma$ of an independent set ($\sigma$ can be
thought of an $n$-dimensional 0/1 vector which is 1 for all the vertices which are in the
independent set and 0 otherwise), we choose a vertex randomly uniformly, if this vertex
is already present in the independent set we keep it with probability $\frac{\lambda}{1 + \lambda}$
and discard it with probability  $\frac{1}{1 + \lambda}$ otherwise the vertex is not in the independent
set and if this vertex can be added to the independent set (i.e. none of its neighbors are already present
in the independent set) then again it is added with with probability $\frac{\lambda}{1 + \lambda}$ and 
rejected with probability  $\frac{1}{1 + \lambda}$. The beauty of this Markov Chain is that the 
stationary distribution is the Gibbs distribution. The details of this chain can be found in \cite{V99}. 

The key difference between the mixing time of Glauber Dynamics for independent
sets and our case is that one can achieve fast mixing time in the former case only 
for small values of $\lambda$. In fact intuitively one should not be able to obtain
fast mixing times for large values of $\lambda$ because such a result would imply that
we can design a randomized polynomial time algorithm for finding maximum independent set
in a graph, which is an NP-Hard problem.  This intuition has led to a series of papers \cite{GGSVY11,MWW09, RSVVY10, S10, W06}
which ultimately has been successful in proving that there exists a threshold value
of  $\lambda = \lambda_c$ such that if $\lambda > \lambda_c$ then estimating the
partition function is hard (the exact technical condition is that unless $NP = RP$ no
FPRAS exists for estimating ${\cal Z}$) and for $\lambda < \lambda_c $ one can obtain an FPTAS for the 
same problem. Previous to this result, the computational complexity of estimating 
the partition function (and counting the number of independent sets)
was only understood for special graphs \cite{DFJ01, V01}.  Apart from these results
substantial attention has been given to obtain good bounds on the mixing time of
this chain for trees \cite{RSVVY10}. 

We define a Markov Chain which is tuned to our need. In our case $\sigma$ is a set 
of edges which form a partial matching. We make a simple modification 
to the above chain, wherein instead of picking a random vertex we pick a random 
edge $e_r \in E$ and perform the same experiment with the parameter $\lambda$ 
as in the case of independent sets (notice that the chosen edge won't be added if in the present matching there is an edge
sharing an end point with $e_r$). We then show that this chain is aperiodic and
irreducible with stationary as the Gibbs distribution over the space of all 
possible partial matchings with parameter $\lambda$. We then use the techniques of bounding
the mixing time to achieve a $\lambda$ independent upper bound. This remarkable
property allows us to exploit the nature of Gibbs distribution (which we obtain
for very large values of $\lambda$) without getting an overhead on the mixing time.

\subsection{Markov Chain Monte Carlo} \label{b}
Markov Chain Monte Carlo algorithms have played a significant role in statistics, econometrics, physics and computing science over the last two decades. 
For some high-dimensional problems in geometry, such as computing the volume
of a convex body in $d$ dimensions, MCMC simulation is the only known general
approach for providing a solution within time polynomial in $d$ \cite{DFK91}. For a number of other hard problems like
approximating the permanent \cite{JS89}, approximate counting \cite{JS96},
the only known FPRASs (\emph{Fully Polynomial time Randomized Approximation Schemes})
rely on the MCMC paradigm. In this paper, we use this method to obtain a faster 
algorithm for the classical problem of finding maximum matchings in general graphs, a problem
which is known to be solvable in polynomial time. \\
The Markov Chain Monte Carlo (MCMC) method is a simple and frequently used
approach for sampling from the Gibbs distribution of a statistical mechanical system. 
The idea goes like this, we design a Markov chain whose state space is $\Omega$
whose stationary distribution is the desired Gibbs distribution.
Starting at an arbitrary state, we simulate the Markov chain on $\Omega$ 
until it is sufficiently close to its stationary distribution. We then output the final state which is a sample from (close to) the desired distribution. The
required length of the simulation, in order to get close to the stationary distribution, is
traditionally referred to as the mixing time $\tau$ or $T_{mix}$ and the aim is to bound
the mixing time to ensure that the simulation is efficient. For a detailed 
understanding of the theory of Markov Chains we would recommend the recent excellent book 
by Levin, Peres and Wilmer \cite{LPW}.    

\section{Preliminaries} \label{Pre}
\subsection{Markov chains}
Consider a stochastic process $(X _t)_{t=0}^\infty$ on a finite state space $\Omega$. Let P denote a
non-negative matrix of size $|\Omega| \times |\Omega|$ which satisfies 
\[
\sum_{j\epsilon\Omega}P_{ij}=1 \mbox {     } \forall  i \in \Omega
\] 

The process is called a Markov chain if for all times $t$ and $i,j \in \Omega$
probability of going from $i^{th}$ state to $j^{th}$ state is independent of the path by which $i^{th}$ state is reached i.e. 
if $X_t$ is the state of the process at time $t$ then
\[
\mbox{P}[X_{t+1} | X_t = x_t, X_{t-1} = x_{t-1} \dots X_0 = x_0] = \mbox{P}[X_{t+1} | X_t = x_t]
\]

A distribution $\pi$ is called a \emph{stationary distribution} if it satisfies
$\pi P = \pi$. A necessary and sufficient condition for a chain to have a unique stationary
distribution is that the chain is

\begin{itemize}

\item[1.] \emph{Irreducible}: for all $i,j \in \Omega$ there exists a time $t$ such that $P^t_{ij} >$ 0; and
\item[2.] \emph{Aperiodic}: for all $i \in \Omega$, GCD \{$t : P^t_{ii} >$ 0\} = 1.

\end{itemize}

A Markov Chain which has both of the above properties is called \emph{ergodic}. 
For an ergodic Markov chain, if a distribution $\pi$ satisfies the detailed balance equations
\[
 \pi_iP_{ij} = \pi_jP_{ji} 
\] 
for all $i,j \in \Omega$ then 
$\pi$ is the (unique) stationary distribution and such a chain is called \emph{reversible}.

\subsection{Mixing Time} \label{c}
The notion of mixing time is defined as a way to measure the closeness of the distribution after
$t$ steps w.r.t. the stationary distribution. The total variation distance between two discrete probability
distributions over a finite space $\Omega$ is defined as the half of the $l_1$ norm of the corresponding 
probability vectors.
\[
d_{TV}(\mu, \nu) = \frac{1}{2} \sum_{\omega \in \Omega} |\mu(\omega) - \nu(\omega)|
\]
If $P_t$ is the probability distribution after $t$ steps in the random walk then $T_{mix}$ is the
minimum $t$ for which,
\[
d_{TV}(P_t, \pi) \leq \frac{1}{2e}
\]
where $\pi$ is the stationary distribution. Therefore, if we intend to get close to a stationary distribution
we just have to truncate the random walk on the state space after $\tau$ steps.

\subsection{Conductance} \label{d}
The conductance of a Markov chain is defined as the following quantity,
\[
\phi(G) = \displaystyle \min_{S \subset V} \frac{\sum_{i \in S, j \in \bar{S}} \pi_i p_{ij}}{(\sum_{i \in S} \pi_i) (\sum_{i \in \bar{S}} \pi_i)}
\]
Another quantity of our interest here is $T_{relax}$, the \emph{relaxation time} of the Markov
chain. $T_{relax}$ is defined as the inverse of conductance i.e. $T_{relax} = \frac{1}{\phi(G)}$. It is 
known that $T_{mix}$ and $T_{relax}$ obey the following inequality \cite{B09}.
\[
T_{relax} + 1\leq T_{mix}.
\]
Also it is known \cite{RSVVY10} that
 \[
 \displaystyle T_{relax} = \Omega \left(\frac{1}{\phi}\right), 
 \]
 a bound usually used to prove lower bounds on the Mixing time of Markov Chains. 

\section{The Chain} \label{e}
In this section, we describe the chain considered by our algorithm which is 
essentially a modification of the hard-core model of Glauber Dynamics. Recall that
our objective is to come up with a chain whose stationary distribution ensures
that the probability of being at a maximum matching is the largest. Recall that first
we need to ensure that our chain is aperiodic and irreducible.

We use the following natural modification of the Markov chain for the hard-core 
model of Glauber dynamics.

\begin{itemize}
\item  Choose an edge $e_r$ uniformly at random from $E$.
\item Let 
\begin{center}
$\sigma' = \begin{cases} \sigma\bigcup\{e_r\}, & \text{with probability} \frac{\lambda}{1+\lambda} \\ \sigma\backslash\{e_r\}, &\text{with probability} \frac{1}{1+\lambda} \end{cases}$
\end{center}
\item If $\sigma'$ is a valid matching, move to state $\sigma'$ otherwise remain at state $\sigma$.
\end{itemize} 

We are now prepared to prove that this is a valid Markov Chain with stationary as the Gibbs distribution, where 
we define Gibbs distribution as the following distribution over the set of all possible matchings $\Omega$ as
\[
 \displaystyle {\cal G}(M) = \frac{\lambda^{|M|}}{ \sum_{x \in  \Omega} \lambda^{|x|} }   
 \]
In the rest, we will call $Z = \sum_{x \in  \Omega} \lambda^{|x|}$.

\begin{lem}
$X_t$ is ergodic with stationary distribution as $\cal G$. 
\end{lem}
\begin{proof}
Since for every state there is some probability by which the walk can remain in the 
same state, the chain is aperiodic. Also the underlying graph is connected because 
given any matching there is a at least one path to reach any other matching (consider the 
path that first drops all the edges of the initial matching and adds the edges of the new
matching one by one). \\
To show that the stationary of this distribution is $\cal G$, we would show that the chain is reversible w.r.t.
$\cal G$. Consider two distinct states $i$ and $j$ in the Markov Chain. Assume w.l.o.g that there
is an edge from $i$ to $j$ (if there is no edge from $i$ to $j$ then the balance equations corresponding
to $i$ and $j$ are trivially satisfied). Let the configuration $i$ has $t$ edges in it then according to 
our construction, $j$ will have either $t+1$ or $t-1$ states (it won't be $t$ because we have assumed that
the two states are distinct). We will just look at the case when $j$ has $t+1$ edges, the other case is analogous.
Since we are interested in showing reversibility w.r.t. $\cal G$, $\pi_i = \frac{\lambda^t}{Z}$ and $\pi_j = \frac{\lambda^{t+1}}{Z}$. Therefore,
 $\pi_i P_{ij} = \frac{\lambda^t}{Z} \cdot \frac{\lambda}{m (1 + \lambda)}$,
and $\pi_j P_{ji} =  \frac{\lambda^{t+1}}{Z} \cdot \frac{1}{m (1 + \lambda)}$. Thus $\pi_i P_{ij} = \pi_j P_{ji} =  \frac{\lambda^{t+1}}{mZ(1+\lambda)}$

\end{proof}

\section{Upper Bound using Coupling } \label{g}
In this section we prove an upper bound on $T_{mix}$ of the chain defined in the previous section. Our
bound is based on the coupling argument introduced by Bubley and Dyer \cite{BD97}. We first give a description
of the coupling method. 

\subsection{ Coupling Method }
A coupling of a Markov chain on state space 
 is a stochastic process $(\sigma_t,\eta_t)$ on
 $|\Omega| \times |\Omega|$ such that:
 \begin{itemize}
 \item  $\sigma_t$ and $\eta_t$ are copies of the original Markov chain and
 \item if $\sigma_t = \eta_t$, then $\sigma_{t+1} = \eta_{t+1}$
\end{itemize}
Thus the chains follow each other after the first instant when
they hit each other. In order to measure the distance between the two
copies of the chain, one introduces a distance function $\Phi$ on the product state space $\Omega \times \Omega$ so that
$\Phi= \Phi(\sigma_t,\eta_t) = 0 \Longleftrightarrow \sigma_t = \eta_t$. For two states $\sigma$ and $\eta$ let
$\rho(\sigma,\eta)$ be the set of all paths from $\sigma$ to $\eta$ in the Markov Chain. The following theorem due to 
to Bubley and Dyer is used to prove mixing time on Markov chain. 
\newline
\begin{thm}\label{bd97} 
Let $\Phi$ be an integer-valued metric defined on  $\Omega \times \Omega$ which takes values in \{0,1...D\} such that, for all $\sigma,\eta \in \Omega$ there exists a path $\xi\in\rho(\sigma,\eta)$ with 
\begin{eqnarray*}\Phi(\sigma,\eta)= \sum_{i}\Phi(\xi^i,\xi^{i+1})
\end{eqnarray*}
Suppose there exists a constant $\beta<$1 and a coupling $(\sigma_t,\eta_t)$ of the Markov chain such that, for all $\sigma_t, \eta_t$,
\begin{eqnarray*}E[\Phi(\sigma_{t+1},\eta_{t+1})]\le\beta\Phi(\sigma_t,\eta_t)
\end{eqnarray*}
Then the mixing time is bounded by\begin{eqnarray*}\tau\le \frac{\log(2eD)}{1-\beta}\end{eqnarray*}
\begin{proof}
Can be found in \cite{BD97,V99}.
\end{proof}
\end{thm}

In order to bound the mixing time, we will  define a coupling so as to
minimize the time until both copies of the Markov chain reach the same state, 
and we will do that by defining the coupling in such a way that on every step both markov chains reach towards same state. 
The aim is to prove a good upper bound on $E[\Phi(\sigma_{t+1},\eta_{t+1})]$ in terms of $\Phi(\sigma_t, \eta_t)$.
In the following subsection we define the coupling:

\subsubsection{Coupling}
Consider the following process $(\sigma_t, \eta_t)$ on $|\Omega| \times |\Omega|$ where $\Omega$ is
the space of all possible matchings.
\begin{defn}
Choose an edge uniformly at random,
\begin{itemize}
\item[1.]  If insertion is possible in both $\sigma$ and $\eta$ then add it probability $\frac{\lambda}{1 + \lambda}$ and remove it with probability
$\frac{1}{1 + \lambda}$.
\item[2.]  If insertion is possible in one and not possible in the other then remove that edge if it
is already present in one matching.
\end{itemize}
\end{defn}
Notice that here we are relying on the fact that one can insert an edge if it is already
present in it. It is easy to verify that this indeed is a coupling for the Markov chain defined in the previous section.

We use the following distance function for the aforementioned coupling. Let
\[
\sigma \oplus \eta =\left\{ e\in E | e \in \left( (\sigma \setminus \eta) \bigcup (\eta \setminus \sigma)\right) \right\} \mbox{ where }\sigma \mbox{ and }\eta \in \Omega
\]
define the distance function $\Phi(\sigma_t, \eta_t) = |\sigma \oplus \eta|$
which is the number of edges present in one but not in the other. Our objective is 
to upper-bound $E[\Phi_(\sigma_{t+1}, \eta_{t+1})]$ in terms of $d_t = \Phi(\sigma_{t}, \eta_{t})$.
Based on the definition of our coupling the following cases may arise once we pick an 
edge $e_r$ uniformly randomly:
\begin{itemize}
\item[1.] \textbf{$e_r$ can be added to both $\sigma_t$ and $\eta_t$}: The subevents are (a) $e_r$
was present in both of them, in this case $\Phi(\sigma_{t+1}, \eta_{t+1}) = d_t$, (b) $e_r$ is not
present in both of them, in which case again $\Phi(\sigma_{t+1}, \eta_{t+1}) = d_t$ and (c) $e_r$
is present in one but not in other, in which case $\Phi(\sigma_{t+1}, \eta_{t+1}) = d_t-1$.

\item[2.] \textbf{$e_r$ can be added in exactly one of $\sigma_t$ and $\eta_t$}: The sub events
for this case are (a) $e_r$ is not present in both, which gives $\Phi(\sigma_{t+1}, \eta_{t+1}) = d_t$
(b)$e_r$ is present in one matching and not in other, in which case $\Phi(\sigma_{t+1}, \eta_{t+1}) = d_t-1$

\item[3.] \textbf{$e_r$ can't be added to any one of $\sigma_t$ and $\eta_t$}: In this case 
$\Phi(\sigma_{t+1}, \eta_{t+1}) = d_t$. 
 
\end{itemize}
Using the above events we can prove the following:

\begin{lem}
$\displaystyle E[\Phi(\sigma_{t+1}, \eta_{t+1})] =\Phi(\sigma_{t}, \eta_{t})\left(1 -\frac{1}{m}\right)$.
\end{lem}
\begin{proof}
Since $\Phi(\sigma_{t+1}, \eta_{t+1})$ can only take two values either $d$ or $d-1$ we only 
need to calculate the the probability of the happening of one of these cases. This happens
when either the event 1(c) or the event 2(b) takes place (as mentioned above). Thus,
\[
\mbox{\textbf {Pr}}[\Phi(\sigma_{t+1}, \eta_{t+1}) = d_t -1] = \mbox{\textbf {Pr}}[1(c) \cup 2(b)]
\]
clearly the distance becomes $d_t -1$ when the chosen edge is one of the edges in $\sigma \oplus \eta$.
We can divide $\sigma \oplus \eta$ in to two sets $U$ and $\bar{U}$. $U$ is the set of edges 
which can be added to the matching in which it is not present, and $\bar{U}$ is the set of 
edge which can't be added to the matching in which it is not present. Using this notation we can write 
the desired probability as 
\begin{eqnarray*}
\displaystyle \mbox{\textbf {Pr}}[1(c) \cup 2(b)] &=& \sum_{e \in U} \frac{1}{m}\left( \frac{\lambda}{1+\lambda} + \frac{1}{1+\lambda} \right)  +  \sum_{e \in \bar{U}} \frac{1}{m} \\
\displaystyle &=& \frac{ |\sigma \oplus \eta|}{m} =  \frac{\Phi(\sigma_t, \eta_t)}{m}
\end{eqnarray*}

Therefore, 
\begin{eqnarray*}
\displaystyle E[\Phi(\sigma_{t+1}, \eta_{t+1})] &=& (d_t -1)\frac{\Phi(\sigma_t, \eta_t)}{m} + d_t \left(1 -\frac{\Phi(\sigma_t, \eta_t)}{m}\right) \\
\displaystyle &=& \Phi(\sigma_{t}, \eta_{t})\left(1 -\frac{1}{m}\right)  (\mbox{ since } d_t = \Phi(\sigma_{t}, \eta_{t}) )
\end{eqnarray*}

\end{proof}

We can now prove the following:
 
\begin{lem}
$T_{mix} = O(m \log n)$.
\end{lem} 
\begin{proof}
Given any $\sigma$ and $\eta$ we define a $d=\Phi(\sigma,\eta)$ length path as $\sigma=\xi^1,\xi^2 \dots ,\xi^d=\eta$ such that $\xi^{i+1}$ is the state obtained by removing exactly one edge from $\xi^i\in \sigma\bigcap(\sigma\oplus\eta)$ for $i=1,2 \dots j$ where $\xi^j$ consists only of edges which do not belong to $\sigma\oplus\eta$ and for all $k\ge j$, $ \xi^{k+1}$ is obtained by adding one edge to $\xi^k$ which belongs to $\eta \bigcap(\sigma\oplus\eta)$.
Since $\xi^i,\xi^{i+1}=1$ for all $i=1,2\dots d-1$, we have 
 \begin{eqnarray*}\Phi(\sigma,\eta)= \sum_{i}\Phi(\xi^i,\xi^{i+1})\\
\end{eqnarray*}
Also we can write 
\[
E[\Phi(\sigma_{t+1}, \eta_{t+1})] =\beta \Phi(\sigma_{t}, \eta_{t})  (\mbox{ with } \beta = \left(1 -\frac{1}{m}\right)) 
\]
this allows us to apply the result from Theorem \ref{bd97} which gives the following result \\
\begin{eqnarray*}
T_{mix}\le \frac{\log 2eD}{\displaystyle 1- \left(1-\frac{1}{m}\right)}\leq m\log(4en)=O(m \log n)\\
\end{eqnarray*}where we have used $\beta  =(1 - 1/m)$ and $D = 2n$.
\end{proof}

Our algorithm is concisely presented as follows: \\

 \begin{algorithm}[H]\label{A}
\KwIn{A Graph $G=(V,E)$ with $|V| = n$ and $|E| = m$}
$\sigma_0  \leftarrow anymatching$\;
$\lambda = 2^m$\;
\For{$t=0$ \emph{\KwTo} $10m \log n$}{
choose an edge $e_r$ uniformly randomly \;
$\sigma' = \begin{cases} \sigma_t\bigcup\{e_r\}, & \text{with probability} \frac{\lambda}{1+\lambda} \\ \sigma_t\backslash\{e_r\}, &\text{with probability} \frac{1}{1+\lambda} \end{cases}$\;
if $\sigma'$ is a valid matching\;
$\sigma_{t+1}=\sigma'$\;
else\;
$\sigma_{t+1}=\sigma_t$\;
}

\KwRet{$\sigma_t$}
\caption{RandMatching\label{RM}}
\end{algorithm}

\begin{thm}\label{prob}
The probability that the random walk in Algorithm \ref{RM} ends on a maximum matching is at least
$\frac{20}{189}$.
\end{thm}
\begin{proof}
Let ${\cal M}_i$ be the set all of matchings of size $i$ (and $S_i$ be its cardinality) in the given graph. Also, let
$k$ be the size of the maximum matching. We need
to find the probability that the random walk lands up on a maximum matching after $T = 10 m \log n \geq T_{mix}$ steps.
Let $X_T$ be the state after $T$ steps, and $\pi_T$ be the probability distribution after $T$ steps then by definition 
of mixing time and triangle inequality, 
\[
\displaystyle  \left|\sum_{\omega \in {\cal M}_k} \pi_T(\omega) -  \sum_{\omega \in {\cal M}_k} {\cal G}(\omega) \right| \leq \sum_{\omega \in {\cal M}_k} |\pi_T(\omega) - {\cal G}(\omega)| \leq   \sum_{\omega \in \Omega} |\pi_T(\omega) - {\cal G}(\omega)| \leq \frac{1}{e} 
\]
where $\sum_{\omega \in {\cal M}_k} \pi_T(\omega) := \mbox{Pr}_k(\pi_T)$ is the probability of reaching a maximum matching after $T$ steps (the success
probability of the algorithm)
and $\sum_{\omega \in {\cal M}_k} {\cal G}(\omega) := \mbox{Pr}_k({\cal G})$ is the probability of finding a maximum matching according to the
Gibbs distribution. We need to find the probability that $X_{T}$ is a maximum matching.  Thus we have,
\[
 \mbox{Pr}_k(\pi_T) \in \left[\mbox{Pr}_k({\cal G})-\frac{1}{e}, \mbox{Pr}_k({\cal G}) + \frac{1}{e} \right]
\]
Also,
\begin{eqnarray*}
\mbox{Pr}_k({\cal G}) &=&\frac{\lambda^kS_k}{\displaystyle\sum_{i=0}^k\lambda^i S_i} \\
&=&\frac{\lambda^kS_k}{\displaystyle\lambda^k S_k + \lambda^{k-1} S_{k-1} + \lambda^{k-2}
 S_{k-2} + \dots +  S_0 }\\
 \end{eqnarray*}
dividing by $\lambda^k$ both numerator and denominator\\
\begin{eqnarray*}
&=&\frac{S_k}{ S_k + \frac{\displaystyle S_{k-1}}{\lambda} + \frac{S_{k-2}}{\lambda^2} + ... +  \frac{S_0}{\lambda^k} }\\
\end{eqnarray*}
We put $\lambda=S_{m}$ where $S_m$=$\displaystyle \max_i^k S_i$\\
\begin{eqnarray*}
&=&\frac{S_k}{ \displaystyle S_k + \frac{ S_{k-1}}{S_m} + \frac{ S_{k-2}}{S_m^2} +
\dots +  \frac{S_0}{S_m^k} } = \frac{1}{\displaystyle1+\frac{S_{k-1}}{S_kS_m} + \frac{ S_{k-2}}{S_kS_m^2} +
\dots +\frac{S_0}{S_kS_m^k}}\\
\end{eqnarray*}
By definition of $S_m$, $\frac{S_i}{S_m}$ is always $\le 1$, hence we have\\
\begin{eqnarray*}
\frac{1}{\displaystyle1+\frac{S_{k-1}}{S_kS_m} + \frac{ S_{k-2}}{S_kS_m^2} +
\dots +\frac{S_0}{S_kS_m^k}} &\ge&\frac{1}{\displaystyle 1 + \frac{1}{S_k} + \frac{ 1}{S_kS_m} +\frac{1}{S_kS_m^2} ... +  \frac{1}{S_kS_m^{k-1}} }\\
&\ge&\frac{1}{\displaystyle 1 + \frac{1}{S_k} \left( 1+\frac{ 1}{S_m} +\frac{ 1}{S_m^2}
\dots +  \frac{1}{S_m^{k-1}}\right) }\\
 &\ge& \frac{1}{\displaystyle 1 +\frac{1}{S_k}\left( \frac{ 1-\left(\frac{1}{S_m}\right)^k}{1-\frac{1}{S_m}}\right) }\\
&\ge&\frac{1}{\displaystyle 1 + \frac{1}{S_k}\left(\frac{S_m^k-1}{S_m^{k-1}(S_m-1)}\right) }\\
\end{eqnarray*}
Since 
\begin{eqnarray*}
\displaystyle \left(\frac{S_m^k-1}{S_m^{k-1}(S_m-1)}\right)&=&\left(1 + \frac{S_m^{k-1} - 1}{S_m^{k-1}(S_m-1)}\right) =1 + \Theta\displaystyle\left(\frac{1}{S_m}\right) \leq \frac{11}{10}
\end{eqnarray*}
\begin{eqnarray*}
\mbox{Pr}_k({\cal G}) &\ge&\frac{1}{ 1 + \frac{1}{\displaystyle \frac{11 S_k}{10}}} \geq \frac{10}{21}\\
\end{eqnarray*}
This gives us,
\[
\mbox{Pr}_k(\pi_T) \geq \frac{10}{21} - \frac{1}{e} \geq \frac{20}{189}
\]
\end{proof}

Notice that our proof still goes through even if we choose a $\lambda$ that is 
larger than $S_m$ (we can take $\lambda = 2^m$). Therefore $\lambda$ can be represented using $m$ bits. 
We can now prove the main theorem,

\begin{thm}
Given a graph $G = (V,E)$ with $|V| = n$ and $|E| = m$, there exists a randomized algorithm that runs
in $O(m \log^2 n)$ time and finds a maximum matching with high probability.
\end{thm}
\begin{proof}
Each step of the algorithm runs in $O(1)$ time (we just need to maintain an array which indicates
whether the $i^{th}$ vertex is occupied in the matching or not), thus
one call of Algorithm \ref{RM} runs in $O(m \log n)$ time which by  Theorem \ref{prob} lands on a maximum matching 
with probability $\frac{20}{169}$. Thus calling it $10 \log n$ times independently, ensures that we land on  a 
maximum matching in one call is at least $1 - \left (\frac{169}{189} \right)^{10 \log n} = 1 - \frac{1}{n^{\Omega(1)}}$.
\end{proof}

\section{Lower Bound via Conductance}
The conductance method as defined in the Section is used to obtain lower bounds on the mixing time 
of the a Markov chain. To get a lower bound on $T_{relax}$ we need an upper-bound on $\phi$, and
by definition of $\phi$, for any cut $(S,\bar{S})$.

\[\phi \leq \displaystyle\left(\frac{\sum_{i \in S,j\in \bar{S}}\pi_iP_{ij}}{(\sum_{i\in S}\pi_i)(\sum_{i\in \bar{S}}\pi_i) } \right)\]

This allows us to observe: \\
\begin{lem}
For any graph $G$, the conductance of our Markov Chain satisfies \[ \phi \leq O\left(\frac{k}{m}\right) \] 
where $k$ is the size of maximum matching.
\end{lem}
\begin{proof}
To give an upper bound on conductance we need to construct a cut $(S, \bar{S})$ for which we can estimate 
the above quantity. Let $S$ be consisting of exactly one matching which is the maximum matching $m_k$
where $k$ is the size of the maximum matching. Thus the number of edges going out of $S$ is $k$.
\[
\displaystyle \frac{\sum_{i \in S,j\in \bar{S}}\pi_iP_{ij}}{(\sum_{i\in S}\pi_i)(\sum_{i\in \bar{S}}\pi_i) } = \frac{\frac{\lambda^k}{Z}\cdot k \cdot \frac{1}{m(\lambda +1)}}{\frac{\lambda^k}{Z} \left(1- \frac{\lambda^k}{Z}\right)} = \frac{k}{(1 + \lambda)m} \cdot \frac{1}{ \left(1- \frac{\lambda^k}{Z}\right)}
\]
Also using the terminology of Theorem \ref{prob}
\begin{eqnarray*}
\displaystyle \frac{\lambda^k}{Z} = \frac{\lambda^k}{\sum_{i=0}^k S_i \lambda^i} = \frac{1}{\sum_{i=0}^k S_{k-i} \frac{1}{\lambda^i}} \leq \frac{1}{\sum_{i=0}^k \frac{1}{\lambda^i}}
= \frac{\lambda^{k}(\lambda-1)}{\lambda^{k+1}-1} = 1 - \Theta\left(\frac{1}{\lambda}\right)
\end{eqnarray*}
Therefore, 
\[\displaystyle \frac{\sum_{i \in S,j\in \bar{S}}\pi_iP_{ij}}{(\sum_{i\in S}\pi_i)(\sum_{i\in \bar{S}}\pi_i) } = \frac{k}{(1 + \lambda)m} \cdot \frac{1}{ \Theta\left(\frac{1}{\lambda}\right)}
= \Theta\left(\frac{k}{m}\right)
\]
Thus the result follows.
\end{proof}
As a result of the previous lemma we have the following result.

\begin{lem}
For any graph $G$, the mixing time of our Markov chain satisfies, $T_{mix} \geq \Omega\left( \frac{m}{k} \right)$
\end{lem}
\begin{proof}
Follows from results in Section \ref{d}.
\end{proof}

Thus the lower bound is sharp (upto logarithmic factors) if the size of the matching in the 
given graph is small (say at most $O(\mbox{poly}\log n)$). \\
 
\textbf{A Note on Other methods to prove Lower Bound}: There are other methods,
apart from conductance, which can be used to prove lower bounds on mixing time. Although
powerful and useful in many contexts, it is not clear whether such methods could be applied
to our chain. For eg. the Wilson's method \cite{B09} expects the knowledge of one eigenvector 
of the transition vector that is different from the all 1's vector and the corresponding eigenvalue
lies in the range $\left(0,\frac{1}{2}\right)$.  Since the matrix $P$ for our case is an exponential
sized matrix with apparently no useful pattern in the entries, it is not clear how to come up 
with such an eigenvector. In fact, we made several educated guesses for coming up with
such a vector all of which failed to serve our purpose.


\section{Conclusion} \label{k}
In this paper, we gave a new randomized algorithm for finding maximum matchings in general 
bipartite graphs that runs in $O(m \log^2 n)$ time that improves upon the running time of Micali and Vazirani.
Our algorithm was based on the MCMC paradigm which performs a truncated random walk on the Markov
Chain defined by the Glauber Dynamics with parameter $\lambda$. Apart from the benefit of being 
very simple (both in analysis and implementation) our algorithm is the first near linear time complexity algorithm for the
maximum matching problem for general graphs. Moreover, unlike \cite{GKK10} the running time of our algorithm is not a random variable. \\  

To our knowledge this is for the first time Glauber dynamics and the nature of Gibbs distribution has been 
exploited to design an faster algorithm for a problem for which efficient solutions are already known, 
and we hope this idea can be of use in other problems as well. 
The obvious open problem will be to improve both the upper-bounds and lower bounds on the mixing time.
Is it possible to improve upon the present bound to get an $O(m)$ time algorithm?  Also 
one can explore the possibility of proving a tighter lower bound of $\Omega(m)$ by using 
more refined techniques. More specifically, it would be interesting to see if one can obtain an 
explicit eigenvector of the transition matrix $P$ (that is different from all 1's vector) and apply Wilson's 
to get an improved lower bound on the mixing time.

\section{Acknowledgements}
Manjish Pal would like to thank Eric Vigoda for answering his mails about the 
Gibbs distribution. Thanks to Purushottam Kar for his feedback on the presentation of the paper.

\end{document}